\documentclass[conference]{IEEEtran}
\ifCLASSINFOpdf
\else
\fi

\usepackage{mathrsfs}
\usepackage{amssymb}
\usepackage{amsmath,bm}

\usepackage{mathrsfs}
\usepackage{graphicx}
\usepackage{eurosym}
\usepackage{amssymb}
\usepackage{amsmath}
\usepackage{amsfonts}
\usepackage{epstopdf}
\usepackage{epsf,subfigure}
\usepackage{psfrag}
\usepackage{graphics}
\usepackage{color} 
\usepackage{cite}
\usepackage{array,multirow,pbox}
\usepackage{enumitem}

\newcommand{\rem}[1]{}
\newtheorem{proposition}{Proposition}

\newtheorem{theorem}{Theorem}

\newtheorem{remark}{Remark}

\newenvironment{proof}[1][Proof]{\begin{trivlist}
\item[\hskip \labelsep {\bfseries #1}]}{\end{trivlist}}

\newcommand{\qed}{\nobreak \ifvmode \relax \else
      \ifdim\lastskip<1.5em \hskip-\lastskip
      \hskip1.5em plus0em minus0.5em \fi \nobreak
      \vrule height0.75em width0.5em depth0.25em\fi}

\graphicspath{{Figures/}}

\newcommand{\mysubeq}[2]{
\begin{subequations}\label{#1}
\begin{align}
#2
\end{align}\end{subequations}}

\begin{document}
%
\title{\vspace{0.25in}{Transient Safety Filter Design for Grid-Forming Inverters}}

\author{\IEEEauthorblockN{Soumya Kundu}
\IEEEauthorblockA{Optimization and Control Group\\
Pacific Northwest National Laboratory, USA\\
Email: soumya.kundu@pnnl.gov}
\vspace{-0.25in}
\and
\IEEEauthorblockN{Karanjit Kalsi}
\IEEEauthorblockA{Optimization and Control Group\\
Pacific Northwest National Laboratory, USA\\
Email: karanjit.kalsi@pnnl.gov}
\vspace{-0.25in}}


%


\maketitle

\begin{abstract}
Unlike conventional generators, inverter-based generation do not possess any rotational inertia. While grid-forming inverters can synthesize small (virtual) inertia via advanced feedback control loops, additional control mechanisms are needed to ensure safety and security of the power grid during transients. In this paper, we propose novel real-time safety-constrained feedback controllers (``safety filters'') for droop-based (grid-forming) inverters to ensure transient security of the grid. The safety filter acts as a buffer between the network operational layer and the inverter-control layer, and only lets those dispatch control signals pass to the inverter droop-controller, which are guaranteed to not violate the safety specifications (frequency, voltage, current limits). Using a distributed barrier certificates method, we construct state-inclusive bounds on the allowable control inputs, which guarantee the satisfaction of transient safety specifications. Sum-of-square programming is used to synthesize the safety filters. Numerical simulation results are provided to illustrate the performance of the proposed filter in inverter-based microgrids.

\end{abstract}


%
\IEEEpeerreviewmaketitle

\section{Introduction}

Power systems operations today are facing a paradigmatic shift as solar, wind and other inverter-based energy resources continue to grow displacing the fossil fueled generators. One one hand, the distributed nature of these inverter-based resources is driving increased deployment of standalone microgrids \cite{pogaku2007modeling}, while on the other hand reliability concerns related to reducing system inertia call for novel control strategies \cite{taylor2016power}. The emerging `grid-forming' inverter technology allows the inverters to behave as a controlled voltage source, as opposed to the conventional controlled current source (or, `grid-following' mode), enabling standalone operation of inverter-based microgrids \cite{lasseter2011smart}, and even synthesizing small (virtual) inertia. However the lower inertia issues still remain in inverter-based microgrids. In particular, resulting shrinking gap in timescales of steady-state (economic) dispatch and real-time control actions invalidates the traditional temporal decoupling between various power grid operations \cite{taylor2016power}. This fact is recognized in various stability and security-constrained optimization formulations for microgrids, as exemplified in \cite{barklund2008energy,Xu:2018,maulik2019stability} and related works. On the other hand, this has also spurred the interest in deriving local and distributed stability conditions for the inverters, specified with respect to its control parameters \cite{simpson2013synchronization,Schiffer:2014,Vorobev:2018,kundu2019identifying,vorobev2019decentralized,nandanoori2020distributed}, to facilitate the emergence of plug-and-play operations in microgrids and distribution networks \cite{farhangi2010path,huang2011future}. Unlike stability and stabilization, the concept of safety violations of node voltages, line currents and frequencies, especially during transients, leading to severe power quality issues and possible damages to the electrical equipment \cite{kusko2007power,Xu:2018} have been largely ignored in the literature. While stabilization refers to steering the power system to (or, close to) its operating point after a disturbance, the concept of safety control implies keeping this system away from poor (and damaging) operating conditions, e.g. avoiding violations of voltage and frequency limits.

The concept of safety control falls under the broader category of constrained control methods, such as the model predictive control \cite{Mayne:2000}, reference governors \cite{garone2017reference,nicotra2018explicit} and barrier functions based methods \cite{Prajna:2007,Wieland:2007}. As opposed to the traditional optimization-based model predictive control and reference governor methods, the work presented in this paper is more aligned with approaches based on the (control) barrier functions (as in \cite{Wieland:2007,Ames:2017,Wang:2018}) and the explicit reference governors \cite{nicotra2018explicit} which guarantee safety via forward invariance of some (safe) set. In particular, such methods allow synthesis of safe feedback control policies by construction of appropriate (control) barrier functions which satisfy certain algebraic conditions guaranteeing set invariance of safety specifications. Recent works have started exploring the applications of barrier functions based methods for transient safety analysis of power systems. In \cite{kundu2019distributed}, the authors proposed a computational algorithm for designing safety certificates for an inverter-based microgrid using sum-of-squares (SOS) algorithms. In \cite{chen2019compositional}, the authors developed an assume-guarantee type contracts for safety of bulk power systems network. Authors in \cite{zhang2019distributed} applied barrier certificates based analysis for the problem of transient safety in terms of hierarchical frequency control in bulk power systems networks.

The main contribution of this paper is in applying barrier functions based method to identify state-inclusive bounds on the control inputs, that can ascertain safety of the network in a distributed manner under uncertainties. In particular, we envision a hierarchical plug-and-play framework of operation for the microgrids (similar to \cite{guerrero2010hierarchical}) where a microgrid coordinator dispatches active and reactive power control set-points to the droop-controlled inverters. The droop-controlled inverters in turn communicate to the microgrid coordinator, a set of bounds on the control inputs, which guarantee safety under fluctuations in the network conditions. Closed-form expressions are provided to synthesize these safety bounds, which act like a ``safety filter'' that allows only the dispatched set-points that lie within those bounds. Numerical results are provided to illustrate the use of the safety filters. The rest of the article is organized as follows: Section\,\ref{S:problem} explains the microgrid model and the describes the problem; Section\,\ref{S:back} presents the necessary background on barrier functions and safety certificates. The main computational and algorithmic developments are described in Section\,\ref{S:algo}, with numerical results presented in Section\,\ref{S:resul}. We conclude the article in Section VI. Throughout the text, $\left|\,\cdot\,\right|$ is used to denote both the Euclidean norm of a vector and the absolute value of a scalar; and $\mathbb{R}\left[x\right]$ to denote the ring of all polynomials in $x\in\mathbb{R}^n$.

\section{Problem Description}\label{S:problem}

\subsection{Microgrid Model}\label{S:model}
We consider the following model of droop-controlled grid-forming inverter dynamics \cite{Coelho:2002,Schiffer:2014,kundu2019distributed}:
\mysubeq{E:droop}{
\dot{\theta}_i & = \omega_i\,,\\
\tau_i\dot{\omega}_i & = -\omega_i + \lambda_i^p \left(P_i^{\text{set}}-P_i\right)\\
\tau_i\dot{v}_i & = v_i^0-v_i + \lambda_i^q \left(Q_i^{\text{set}}-Q_i\right)
}
where $\lambda_i^p>0$ and $\lambda_i^q>0$ are the droop-coefficients associated with the active power vs. frequency and the reactive power vs. voltage droop curves, respectively; $\tau_i$ is the time-constant of a low-pass filter used for the active and reactive power measurements; $\theta_i\,,\,\omega_i$ and $v_i$ are, respectively, the phase angle, speed and voltage magnitude; $v^0_i$ is the desired (nomial) voltage magnitude; $P_i^{\text{set}}$ and $Q^{\text{set}}_i$ are the active power and reactive power set-points, respectively. Finally, $P_i$ and $Q_i$ are, respectively, the active and reactive power injected into the network which relate to the neighboring bus voltage phase angle and magnitudes as:
\mysubeq{E:PQ}{
P_i &= v_i{\sum}_{k\in\mathcal{N}_i} v_k\left(G_{i,k}\cos\theta_{i,k} + B_{i,k}\sin\theta_{i,k}\right)\\
Q_i &= v_i{\sum}_{k\in\mathcal{N}_i} v_k\left(G_{i,k}\sin\theta_{i,k} - B_{i,k}\cos\theta_{i,k}\right)
}
where $\theta_{i,k}=\theta_i-\theta_k$\,, and $\mathcal{N}_i$ is the set of neighbor nodes. $G_{i,k}$ and $B_{i,k}$ are respectively the transfer conductance and susceptance values of the line connecting the nodes $i$ and $k$\,. At the equilibrium (steady-state) operation:
\begin{align*}
\forall i:\quad P_i=P_i^{\text{set}},\,Q_i=Q_i^{\text{set}},\,\omega_i=0,\,v_i=v_i^0\,.
\end{align*} 
The dynamics of the inverters \eqref{E:droop} bear some similarities with the classical swing equation-based models of synchronous generators, where the time-constant $\tau$ resembles the rotational inertial mass of the conventional generators. The values of the time-constants are typically orders of magnitude smaller than the rotational inertial mass. As such, these low-inertial inverter-based systems are more prone to transient fluctuations than bulk power grid  \cite{Xu:2018}, requiring faster time-scale control of the microgrid to ensure operational reliability.

In a hierarchical microgrid control architecture (illustrated in Fig.\,\ref{F:safety_multi}), a microgrid operator (or, coordinator) would be dispatching the active and reactive power setpoints to the individual inverters, to achieve certain operational goals (e.g. loss minimization). The dispatched active and reactive power set-points can be modeled as:
\begin{align}
P_i^{\text{set}} = P^0_i+u^p_i\,,~Q_i^{\text{set}} = Q^0_i+u^q_i\,,
\end{align}
where $P^0_i$ and $Q^0_i$ are the set-points for the unperturbed (or nominal) operating condition; and $u^p_i$ and $u^q_i$ are any deviations from the nominal. Changes in the dispatched control set-points, in addition to uncertainties and variability in the network, drive transient fluctuations in the terminal voltage and frequency which may violate the `safety' limits determined via engineering design, such as:
\begin{align*}
\underline{v_i}\leq v_i(t)\leq \overline{v_i}\,,\quad \underline{\omega_i}\leq \omega_i(t)\leq \overline{\omega_i}\,.
\end{align*}

In this paper, we are interested in synthesizing safe set of possible dispatched set-points of the form:
\begin{align*}
u_i^p\in\mathcal{U}_i^p\,,\quad u_i^q\in\mathcal{U}_i^q
\end{align*}
such that safety of the voltage and frequency at the terminal of each inverter is guaranteed robustly under a set of (uncertain) operating conditions. 

\subsection{Safety Filter: Problem Formulation}
\begin{figure}[thpb]
\centering
\includegraphics[scale=0.154]{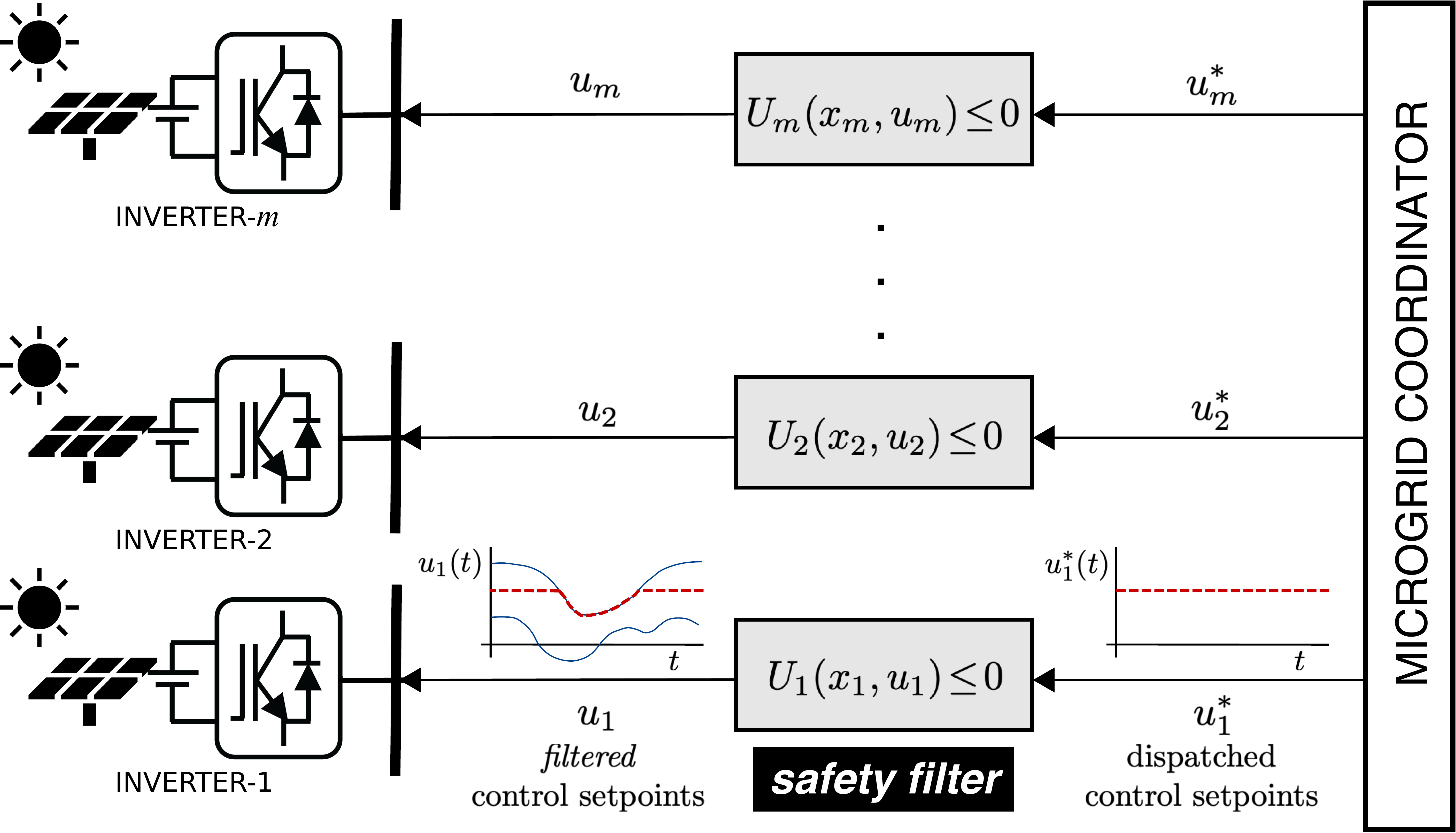}
\caption{Illustration of the safety filters acting at the individual inverter terminals blocking out (and modifying) any coordinator dispatched set-points that are (potentially) unsafe.}
\label{F:safety_multi}
\end{figure}
In this paper, we are not interested in computing any particular safety control policies. Instead, we are interested in identifying a set of allowable control inputs which will ensure safety of the sub-system under bounded disturbances from the neighbors. In particular, as illustrated in Fig.\,\ref{F:safety_multi}, we would like to identify state-inclusive semi-algebraic set of the form:
\begin{align}\label{E:sector_idea}
\mathcal{U}_i(x_i):=\left\lbrace u\,\left| \, U_i(x_i,u)\leq 0\right.\right\rbrace
\end{align}
where $U_i(\cdot)$ is a set of polynomials in $x_i$ and $u$\,, such that the safety of the sub-system is guaranteed for every control input $u_i\in\mathcal{U}_i(x_i)$ whenever the disturbances from the neighbor are norm-bounded. In this particular case, we will use the barrier level-sets ($B_j(x_j)\geq c_j$), introduced in Sec.\,\ref{S:back}, as a proxy for the norm-bounds on the neighbor states ($x_j$), but the results hold for any other types of norm-bounds.

\section{Background}\label{S:back}

\subsection{Safety Certificates: Barrier Functions}

Consider a nonlinear dynamical system of the form 
\begin{align}\label{E:f}
\dot{x}(t) &= f(x(t))~ \forall t \geq 0\,,~x\in\mathbb{R}^n\,,
\end{align}
with an equilibrium at the origin $(f(0)= 0)$, where $f:\mathbb{R}^n \rightarrow \mathbb{R}^n$ is locally Lipschitz. For brevity, we would drop the argument $t$ from the state variables, whenever obvious. 
In contrast to asymptotic stability which concerns with the convergence of the state variables to the stable equilibrium, the notion of `safety' comes from engineering design specifications. From the design perspective, the system trajectories are not supposed to cross into the certain regions in the state-space marked as `unsafe'. Let us assume that the `unsafe' region of operation for the system \eqref{E:f} is given by the domain 
\begin{align}
\mathcal{X}_u:=\lbrace x\in\mathbb{R}^n\left\vert \,w_i(x)> 0\,,~i=1,2,\dots,l\right.\rbrace
\end{align}
where $w_i:\mathbb{R}^n\mapsto\mathbb{R}$ are a set of $l$ ($\geq 1$) polynomials. Safety of such systems can be verified through the existence (or, construction) of continuously differentiable barrier functions $B:\mathbb{R}^n\mapsto\mathbb{R}$ of the form \cite{Prajna:2007,Wieland:2007,Ames:2017,Wang:2018}:
\begin{subequations}\label{E:B}
\begin{align}
B(x)&\geq 0\quad \forall x\in\mathbb{R}^n\backslash\mathcal{X}_u\\
B(x)&<0\quad \forall x\in\mathcal{X}_u\\
 (\nabla_x B)^T\!f(x)+\alpha\left(B(x)\right) &\geq 0\quad \forall x\in\mathbb{R}^n
\end{align}
\end{subequations}
where $\alpha(\cdot)$ is an extended class-$\mathcal{K}$ function\footnote{A continuous function $\alpha:(-a,b)\mapsto(-\infty,\infty)$\,, for some $a,b>0\,,$ is \textit{extendend class-$\mathcal{K}$} if it is strictly increasing and $\alpha(0)=0$ \cite{Khalil:1996}.}. The third condition ensures that at the level-set $B=0$ the value of the barrier function is increasing along the system trajectories. Safety is guaranteed for all trajectories starting inside the domain $\lbrace x\left|\,B(x)\geq 0\right.\rbrace$ which is \textit{invariant} under the dynamics \eqref{E:f}.

\subsection{Sum-of-Squares (SOS) Optimization}
A multivariate polynomial $p\!\in\!\mathbb{R}\left[x\right],~x\!\in\!\mathbb{R}^n$, is a \textit{sum-of-squares} (SOS) polynomial if there exist some polynomial functions $h_i(x), i = 1\ldots s$ such that 
$p(x) = \sum_{i=1}^s h_i^2(x)$.
We denote the ring of all SOS polynomials in $x$ by $\Sigma[x]$.
The problem of determining whether or not a given polynomial is an SOS can be cast into an equivalent semi-definite problem \cite{sostools13}. 
An important result from algebraic geometry, called Putinar's Positivstellensatz theorem \cite{Putinar:1993,Lasserre:2009}, helps in translating conditions such as in \eqref{E:B} into SOS feasibility problems. 
\begin{theorem}\label{T:Putinar}
Let $\mathcal{K}\!\!=\! \left\lbrace x\in\mathbb{R}^n\left\vert\, k_1(x) \geq 0\,, \dots , k_m(x)\geq 0\!\right.\right\rbrace$ be a compact set, where $k_j$ are polynomials. Define $k_0=1\,.$ Suppose there exists a $\mu\!\in\! \left\lbrace {\sum}_{j=0}^m\sigma_jk_j \left\vert\, \sigma_j \!\in\! \Sigma[x]\,\forall j \right. \right\rbrace$ such that $\left\lbrace \left. x\in\mathbb{R}^n \right\vert\, \mu(x)\geq 0 \right\rbrace$ is compact. Then,  
\begin{align*}
p(x)\!>\!0~\forall x\!\in\!\mathcal{K}
\!\implies\! p \!\in\! \left\lbrace {\sum}_{j=0}^m\sigma_jk_j \left\vert\, \sigma_j \!\in\! \Sigma[x]\,\forall j \right. \right\rbrace\!.
\end{align*}
\end{theorem}

\begin{remark}\label{R:Putinar}
Using Theorem\,\ref{T:Putinar}, one can translate the problem of checking that $p\!>\!0$ on $\mathcal{K}$ into an SOS feasibility problem where we seek the SOS polynomials $\sigma_0\,,\,\sigma_j\,\forall j$ such that $p\!-\!\sum_j\sigma_j k_j$ is {SOS.
Note that any} equality constraint $k_i(x)\!=\!0$ can be expressed as two inequalities $k_i(x)\!\geq 0$ and $k_i(x)\!\leq\! 0$. In many cases, especially for the $k_i\,\forall i$ used throughout this work, a $\mu$ satisfying the conditions in Theorem\,\ref{T:Putinar} is guaranteed to exist (see \cite{Lasserre:2009}), and need not be searched for.
\end{remark}

\subsection{Distributed Barrier Certificates}
An interconnected network model of the microgrid with $m$ droop-controlled inverters can be compactly expressed as:
\mysubeq{E:network}{
\dot{x}_i&=f_i(x_i)+g_i(x_i)u_i + {\sum}_{j\in\mathcal{N}_i}h_{ij}(x_i,x_j)\,,\\
\mathcal{X}_{u,i}&:=\lbrace x_i\,|\,w_j(x_i)\geq 0\,,~j=1,2,\dots,l_i\rbrace
}
where each $i\in\lbrace 1,2,\dots,m\rbrace$ identifies an inverter. $x_i\in\mathbb{R}^{n_i}$ is the $n_i$-dimensional state vector associated with the $i$-th inverter, while $u_i$ is an $r_i$-dimensional control input vector. We assume that the origin is an operating point of interest of the networked system. We assume that $h_{ij}(x_i,0)=0$ for all $x_i$\,. Moreover $f_i,\,g_i$ and $h_{ij}$ are locally Lipschitz functions. 

In a recent work \cite{kundu2019distributed}, a design procedure was presented to compute state-feedback control policies that guarantee safety of the inverter-based microgrids via distributed barrier certificates as summarized in the following result:

\begin{theorem}
\cite{kundu2019distributed} If there exist continuous functions $B_i(x_i)$\,, control policies $u_i$ and non-negative scalars $c_i$ satisfying 
\mysubeq{E:conditions}{
\forall i:~&B_i(0)>c_i\label{E:inclusion}\\
&B_i(x_i)<0\quad \forall x_i\in\mathcal{X}_{u,i}\label{E:safety}\\
 &\dot{B}_i \geq 0\quad \forall x_i\in\partial\mathcal{D}_i[c_i],\,\forall x_j\in\mathcal{D}_j[c_j]~\forall j\in\mathcal{N}_i\label{E:derivative}\\
 &\dot{B}_i=\nabla_{x_i}B_i^T(f_i(x_i)+g_i(x_i)u_i + {\sum}_{j\in\mathcal{N}_i}h_{ij}(x_i,x_j))\,.\notag
}
where $\mathcal{D}_i[c_i]:=\lbrace x_i\,|\,B_i(x_i)\geq c_i\rbrace~\forall i$ and $\partial\mathcal{D}_i[c_i]:=\lbrace x_i\,|\,B_i(x_i)=c_i\rbrace$ is the boundary set of the domain $\mathcal{D}_i[c_i]$\,. then the safety of the interconnected system \eqref{E:network} is guaranteed for all $t\geq 0$ whenever $B_i(x_i(0))\geq c_i\,\forall i$\,, i.e.
\begin{align*}
x_i(0)\in\mathcal{D}_i[c_i]~\forall i\implies x_i(t)\in\mathbb{R}^{n_i}\backslash\mathcal{X}_{u,i}~\forall i\,\forall t\geq 0\,.
\end{align*}
Moreover there is a neighborhood $\mathcal{X}_i$ around origin (i.e. $0\in\mathcal{X}_i\,\forall i$) such that $\mathcal{X}_i\subseteq\mathcal{D}_i[c_i]$\,.\hfill\hfill\qed
\end{theorem}

SOS based techniques were used in \cite{kundu2019distributed,Wang:2018} to compute the barrier functions for each isolated and autonomous sub-system of the form: $ \dot{x}_i=f_i(x_i)$\,, satisfying the conditions in \eqref{E:B}. The set $\mathcal{D}_i[0]=\lbrace x\,\left|\,B_i(x)\geq 0\right.\rbrace$ is an invariant subset of the safety region of the isolated subsystem $i$\,. Moreover, a feedback control policy that satisfied the safety conditions can be computed by solving an optimization problem similar to:
\mysubeq{E:control_sos}{
\forall i:&\qquad\min_{u_i(x_i)}\quad \bar{u}_i\\
&\text{s.t.}\quad \|u_i(x_i)\|_\infty\leq \bar{u}_i\,~\forall x_i\in\mathcal{D}_i[c]\,,\\
&\text{and }\left\lbrace\begin{array}{l}\nabla_{x_i}B_i^T(f_i+g_iu_i+{\sum}_{j\in\mathcal{N}_i}h_{ij})\geq 0\\
\forall x_i\in\partial\mathcal{D}_i[c_i]\,,\,\forall x_j\in\mathcal{D}_j[c_j],\,j\in\mathcal{N}_i\end{array}\right.\label{E:control}
}

\section{Safety Filter Design: Main Results}\label{S:algo}

Consider the interconnected network described by \eqref{E:network}, where each isolated and autonomous subsystem-$i$ of the form $\dot{x}_i \!=\! f_i(x_i)$ admits a barrier function $B_i(x_i)$, satisfying the conditions in \eqref{E:B}, such that the domain $\mathcal{D}_i[0]$ is an invariant subset of the safety region specified for the isolated subsystem $i$. Moreover, for convenience, each of the subsystem barrier functions are scaled so that $B_i(0)=1\,\forall i$\,. As such, in this section, we will concern ourselves with the bounded region of the interconnected network of the form: 
\begin{align*}
\mathcal{D}_1[c_1]\times\mathcal{D}_2[c_2]\times\dots\times\mathcal{D}_m[c_m]
\end{align*}
for some $c_i\in[0,1)$\,. We can state the following results:
\begin{proposition}\label{P:existence}
Consider subsystem-$i$ of the interconnected system \eqref{E:network}. If $\left|\nabla B_i^Tg_i\right|\!>\!0$ on the boundary set $\partial\mathcal{D}_i[c_i]$\,, then, under bounded disturbances from the neighbors $j\in\mathcal{N}_i$ of the form $x_j\in\mathcal{D}_j[c_j]$, there exists a state-feedback control policy for subsystem-$i$ such that $\mathcal{D}_i[c_i]$ is an invariant domain.
\end{proposition}
\begin{proof}
Let us select $u_i\!=\!\beta_i g_i^T\nabla B_i$\, for some
\begin{align*}
\beta_i\geq \underset{{x_i\in\partial{D}_i[c],x_j\in\mathcal{D}_j[c]}}{\max}\frac{\left|\nabla B_i^T\left(f_i+{\sum}_{j\in\mathcal{N}_i}h_{ij}\right)\right|}{\left|\nabla B_i^Tg_i\right|^2}\,.
\end{align*}
By construction, this $u_i$ satisfies the condition \eqref{E:control}, thereby guaranteeing invariance of the domain $\mathcal{D}_i[c_i]$.\hfill\qed
\end{proof}

\begin{proposition}\label{P:sector}
Consider subsystem-$i$ of the interconnected system \eqref{E:network}. If $u_i^*$ is some state-feedback control policy which guarantees invariance of $\mathcal{D}_i[c_i]$ under bounded disturbances from the neighbors $j\in\mathcal{N}_i$ of the form $x_j\in\mathcal{D}_j[c_j]$, then so does a family of state-feedback control policies of the form $\lbrace u_i^*+\beta g_i^T\nabla B_i\,|\,\beta\!\geq\! 0\rbrace$\,.
\end{proposition}
\begin{proof}
This follows trivially when we notice that for every feedback control policy $u_i^*$ satisfying \eqref{E:control}, any feedback policy $u_i^*+\beta g_i^T\nabla B_i$\,, for arbitrary $\beta\!\geq\! 0$\,, also does so. \hfill\qed
\end{proof}

While the above results are obtained without any consideration of possible bounds on the control inputs, in practice there are finite limits on the control inputs that can be applied. Assuming that there exists a safety feedback control policy $u_i^*$ which satisfies the bounds on the control inputs, we can find a $\beta^{\max}>0$ such the family of safety feedback control policies 
\begin{align*}
\lbrace u_i^*+\beta g_i^T\nabla B_i\,|\,\beta\in[0,\beta^{\max}]\rbrace
\end{align*} 
will generate control inputs satisfying the control bounds. 

Therefore we can construct a state-inclusive semi-algebraic set of the form \eqref{E:sector_idea} with a vector of polynomials:
\begin{align}\label{E:sector}
U_i(x_i,u)&=\left((u_i^*(x_i)\!-\!u\right)\odot\left(u_i^*\!+\!\beta^{\max} g_i^T\nabla B_i\!-\!u\right),
\end{align}
where the symbol $\odot$ represents component-wise multiplication, such that whenever a control input $u_i(t)$ chosen (arbitrarily) from the set $\mathcal{U}_i(x_i(t))$\,, i.e.
\begin{align}
\forall t:\quad u_i(t)\in\mathcal{U}_i(x_i(t))\,,
\end{align} 
the invariance of the domain $\mathcal{D}_i[c_i]$ is guaranteed under bounded disturbances, $x_j\in\mathcal{D}_j[c_j]$\,, from the neighbors. Note that, for polynomial barrier functions that are quadratic or higher order polynomials, we have $\nabla B_i(0)\!=\!0$\,, i.e., by construction, the semi-algebraic set \eqref{E:sector} leads to sector-like bounds on the control input.

%

However, such sector-like bounds are too restrictive. The reason is that the semi-algebraic set theoretic condition \eqref{E:control}, used to construct \eqref{E:sector}, yields a control policy that needs to be activated only close to the boundary of the set defined by $\mathcal{D}_i[c_i]$, a subset of the safety region. As such, a safety control set defined in the form of \eqref{E:sector_idea}, with $U_i$ defined in \eqref{E:sector}, would invariably lead to unnecessarily restricted set of safe control inputs, especially close to the origin. This is circumvented in the following construction of safety control bounds:

\begin{theorem}\label{T:main}
Consider subsystem-$i$ of the network \eqref{E:network}. If $\left|\nabla B_i^Tg_i\right|\!>\!0$ on the boundary set $\partial\mathcal{D}_i[c_i]$\,, then there exist state-feedback control policies $u_i^\alpha(x_i)$ and $u_i^\theta(x_i)$, satisfying $u_i^\alpha(x_i)\!<\! u_i^\theta(x_i)$ (strict vector inequality) everywhere in $\mathcal{D}_i[c_i]$, 
such that the family of state-feedback control policies
\begin{align*}
u_i(x_i)\in\left\lbrace r\,u_i^\alpha(x_i)+(1-r)\,u_i^\theta(x_i)\,\left|\,r\in[0,1]\right.\right\rbrace
\end{align*}
guarantees invariance of $\mathcal{D}_i[c_i]$ under bounded disturbances from the neighbors, $x_j\in\mathcal{D}_j[c_j]\,\forall j\in\mathcal{N}_i$\,.
\end{theorem}
\begin{proof}
Recall that the Propositions\,\ref{P:existence} and \ref{P:sector} combine to show the existence of sector-like bounds for safety control inputs given in the form of \eqref{E:sector_idea} with $U_i(x_i,u)$ given by \eqref{E:sector}. Now, note that we can use the properties of the barrier functions to expand the allowable safety control set as:
\begin{align}\label{E:sector_relax}
\mathcal{U}_i(x_i):=\left\lbrace u\,\left|\, U_i(x_i,u)\!\leq\! \gamma\log\left(\!\frac{1\!-\!c_i}{1\!-\!B_i(x_i)}\!\right)\right.\right\rbrace,
\end{align}
where $\gamma$ is a positive scalar termed as the `relaxation coefficient', and $U_i$ is defined in \eqref{E:sector}. Observe that the right-hand side of the inequality is non-negative inside the domain $\mathcal{D}_i[c_i]$, with it approaching zero at the boundary $\partial\mathcal{D}_i[c_i]$ and infinity at the origin. The quadratic inequality can be solved to obtain a family of safety control inputs,
\begin{align*}
\lbrace u\,|\,u_i^\alpha(x_i)\leq u\leq u_i^\theta(x_i)\rbrace
\end{align*}
guaranteeing the invariance of $\mathcal{D}_i[c_i]$ under bounded neighbor disturbances $x_j\in\mathcal{D}_j[c_j]\,\forall j\in\mathcal{N}_i$\,. Moreover, it can be shown that (detailed omitted), the vector difference $\Delta u(x_i):=u_i^\theta(x_i)-u_i^\alpha(x_i)$\,, is proportional to the component-wise square-root of the term
\begin{align*}
\left(\beta^{\max}\right)^2\left(g_i\nabla B_i\right)^2+4\gamma\log\left(\!\frac{1\!-\!c_i}{1\!-\!B_i(x_i)}\!\right)
\end{align*}
where the first term denotes component-wise square. Second term is zero only at the boundary $\partial\mathcal{D}_i[c_i]$. But, as per assumption, $\left|\nabla B_i^Tg_i\right|\!>\!0$ on the boundary set $\partial\mathcal{D}_i[c_i]$\,. Therefore $\Delta u(x_i)>0$ everywhere inside $\mathcal{D}_i[c_i]$\,. 
\hfill\qed\end{proof}

Note that the safety control set \eqref{E:sector_relax} in Theorem\,\ref{T:main} can be recast into \eqref{E:sector_idea} by modifying the definition of $U_i$ in \eqref{E:sector}. Theorem\,\ref{T:main} provides an algebraic condition that can be easily checked to verify whether a dispatched control input is safe or not, using only locally available information (such as the states belonging to the subsystem-$i$ alone). This contrasts with the other alternative methods, such as solving a semi-algebraic set condition similar to the one in \eqref{E:control} which requires solving an equivalent, and non-trivial, semi-definite optimization, to verify safety guarantees under a control policy, or requires information from neighboring subsystems in order to verify safety under a dispatched control input.

\section{Numerical Results}\label{S:resul}
%
%
We consider the microgrid example used in \cite{kundu2019distributed}, which was modified from the network described in \cite{ersal2011impact}. Specifically, we disconnect the utility and replace the substation (bus 0) by a droop-controlled inverter, with three other inverters placed on buses 1, 4 and 5\,. Inverter dynamics were modeled in the form of \eqref{E:droop}. Bus 0 was considered as the reference bus for the phase angle. The droop coefficients $\lambda^p_i$ and $\lambda^q_i$ were selected as $2.43\,$rad/s/p.u. and $0.2$\,p.u./p.u., respectively; while the filter time-constant $\tau_i$ was set to 0.5\,s \cite{Schiffer:2014}. Nominal values of voltage and frequency, as well as the active and reactive power set-points were obtained by solving the steady-state power-flow equations \eqref{E:PQ}, which were then used to shift the nominal operating point to the origin. The loads were modeled as constant power loads, and a Kron reduced network with only the inverter nodes were used for analysis.
The unsafe set was defined in terms of the shifted (around the 1\,p.u.) nodal voltage magnitudes as $v_i< -0.4\,\text{p.u.}~\text{ or }~\,v_i> 0.2$ p.u.
In this paper, we focus only on the transient voltage safety limits, but other safety limits (such as frequency, phase angle difference, etc.) can be formulated as well.

\begin{figure}[thpb]
\centering
\subfigure[]{
\includegraphics[scale=0.4]{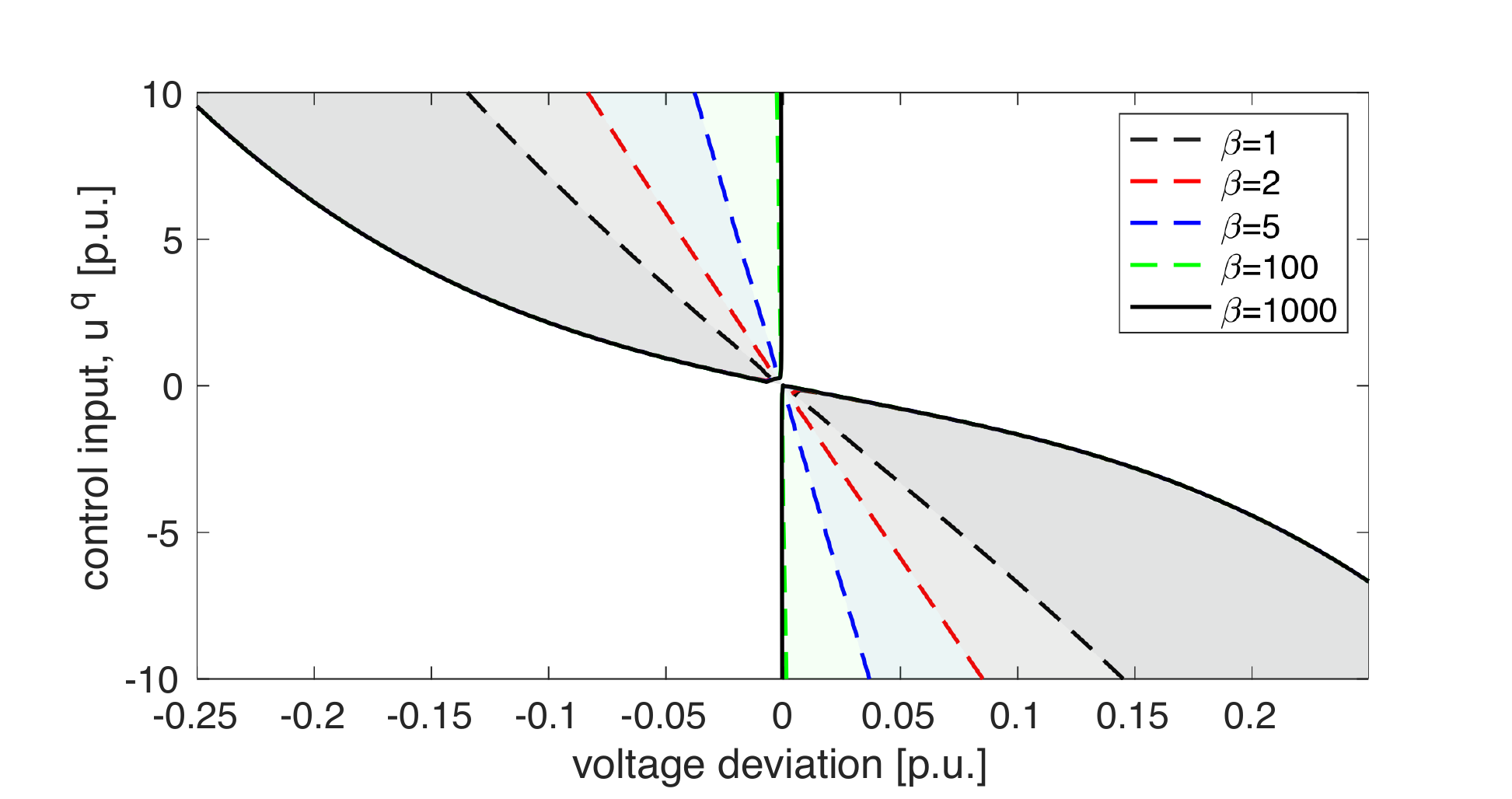}\label{F:beta}
}
\subfigure[]{
\includegraphics[scale=0.4]{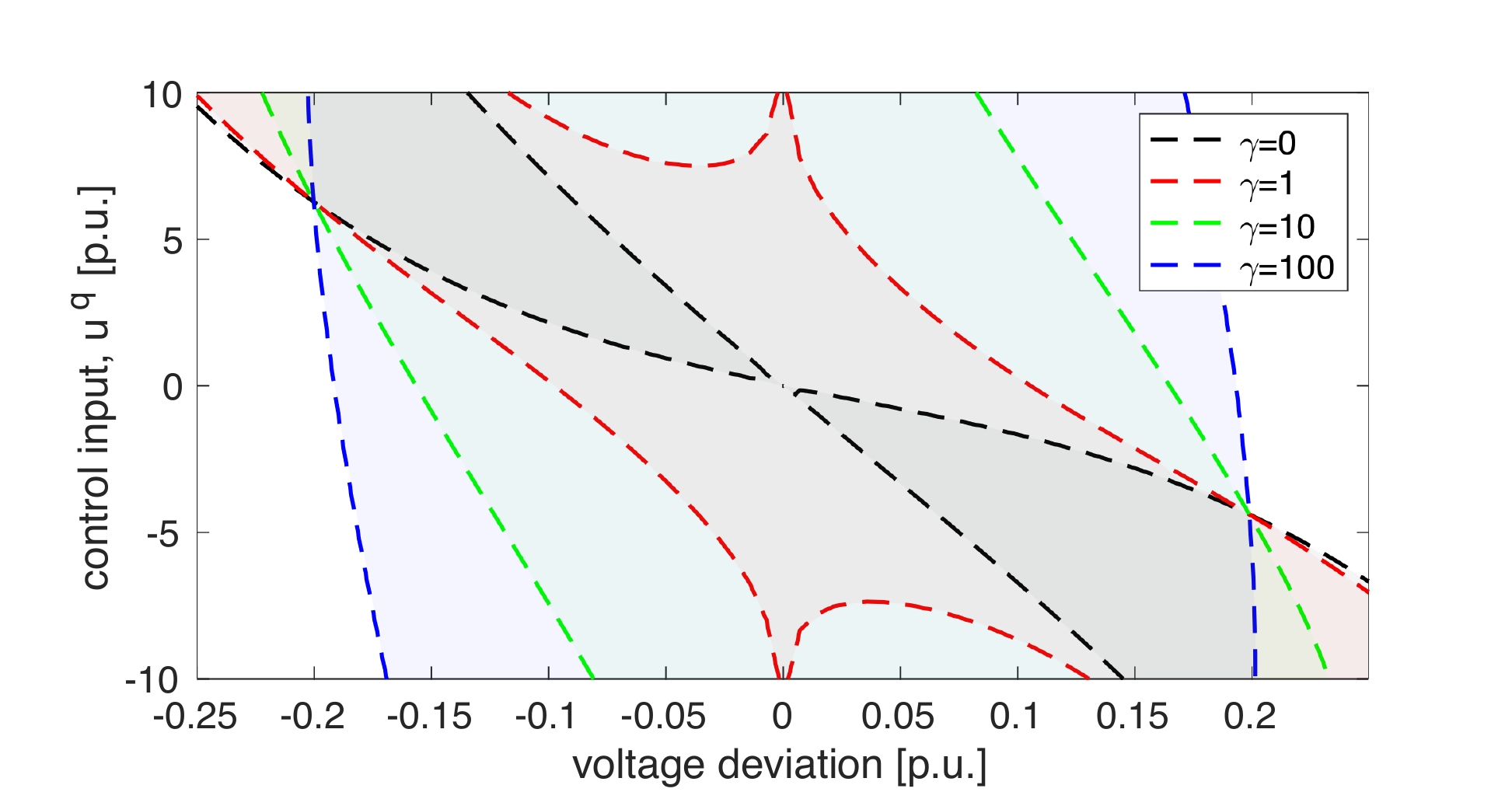}\label{F:gamma}
}
\subfigure[]{
\includegraphics[scale=0.4]{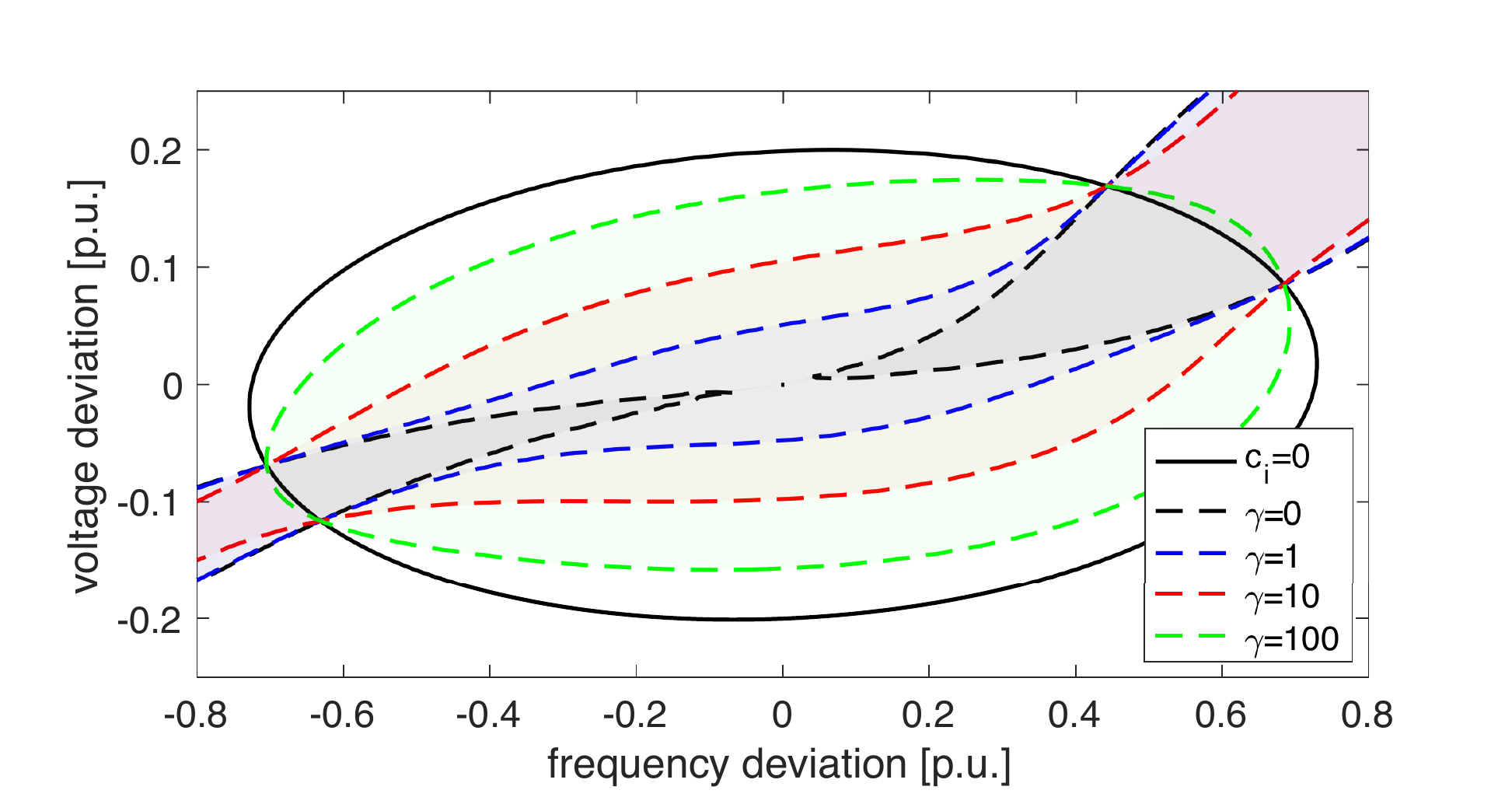}\label{F:2D}
}
\caption[]{Illustration of the state-inclusive bounds on the safety control inputs for the inverter 3: (a) We keep $\gamma=0$ and vary $\beta$. (b) We keep $\beta=1$ and vary $\gamma$. (c) We keep $\beta=1$ and vary $\gamma$, and show at which state values the control input $u^q=0$ is deemed `safe'.}
\label{F:res_sector}
\end{figure}

Barrier functions are computed for each isolated and autonomous system using the algorithm described in \cite{Wang:2018,kundu2019distributed}. An optimization problem similar to \eqref{E:control_sos} is solved to compute a control policy $u_i^*(x)$ that guarantees the invariance of the domains defined by the barrier level-sets $c_i=0\,\forall i$\,. Then the bounds of the control inputs are computed for each subsystem using the closed-form expressions proposed in \eqref{E:sector_relax}. Fig.\,\ref{F:res_sector} illustrates the obtained state-inclusive bounds on the allowable safe control inputs. The first two sub-plots show the safe values of reactive power control input, as a function of the voltage deviation (while frequency deviation is kept at zero, i.e. a projection) for various values of $\beta$ and $\gamma$\,. The third sub-plot shows the volume of the state-space over which the control input $u^q=0$ is deemed to be `safe', for various values of $\gamma$, with $\beta=1$\,. In order to illustrate the effect of the sector-bounds, we show the time-series simulations of two scenarios. The subsystem 3 results are presented in Fig.\,\ref{F:res_time}. Uncertainties are created by randomly sampling the neighboring states $(x_j)$ from the set $\mathcal{D}_j[0]$\,. The top plots shows that without any safety filter, the voltage trajectories violate the safety bounds when the dispatched set-point is $u_i^q=0$. The bottom two plots show the safety filter is active with the control bounds, with $(\beta,\gamma)\!=\!(1,100)$, and the voltage stays within the safety limits.

\begin{figure}[thpb]
\centering
\subfigure[]{
\includegraphics[scale=0.4]{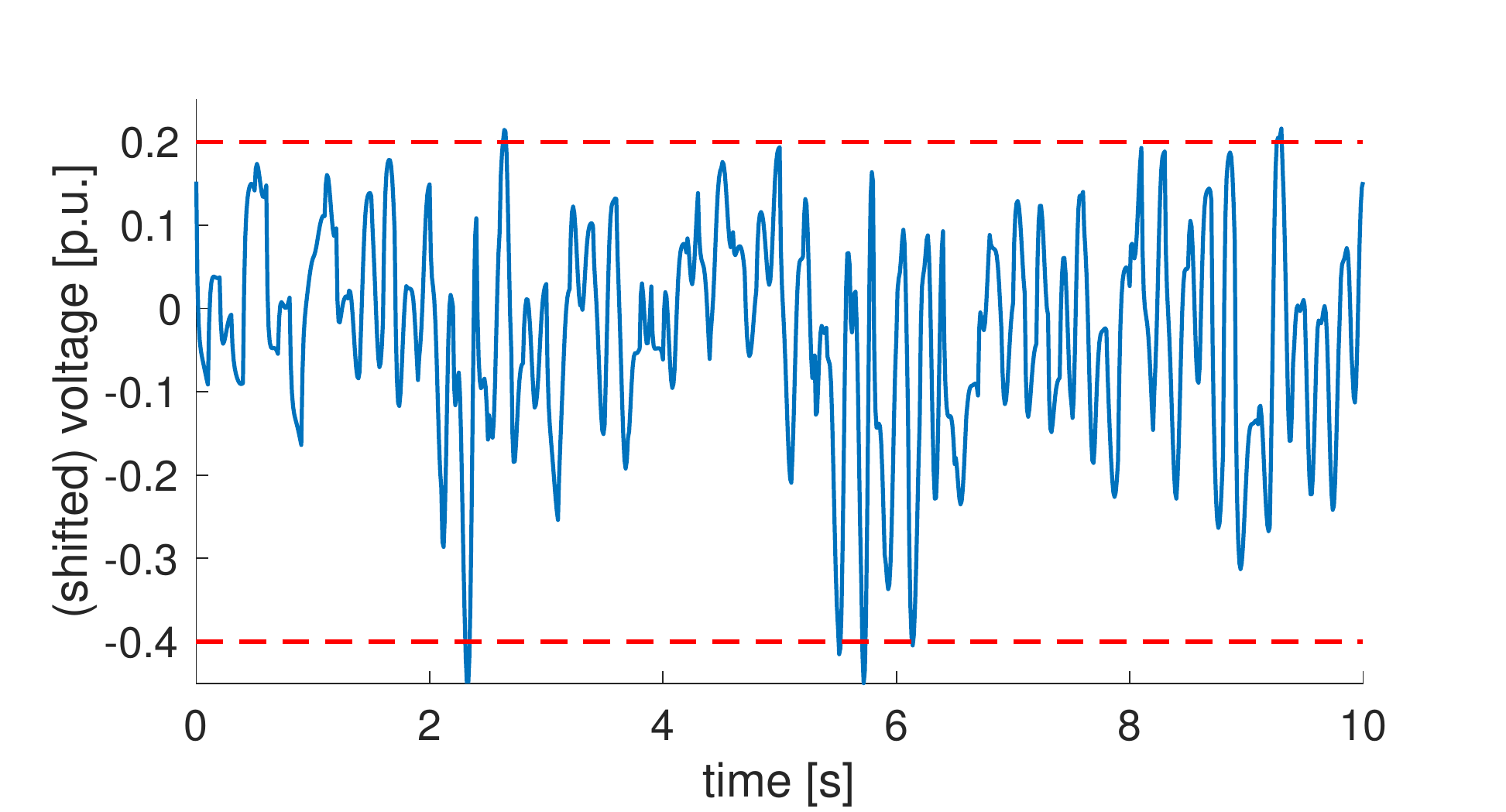}\label{F:res3_v}
}
\subfigure[]{
\includegraphics[scale=0.4]{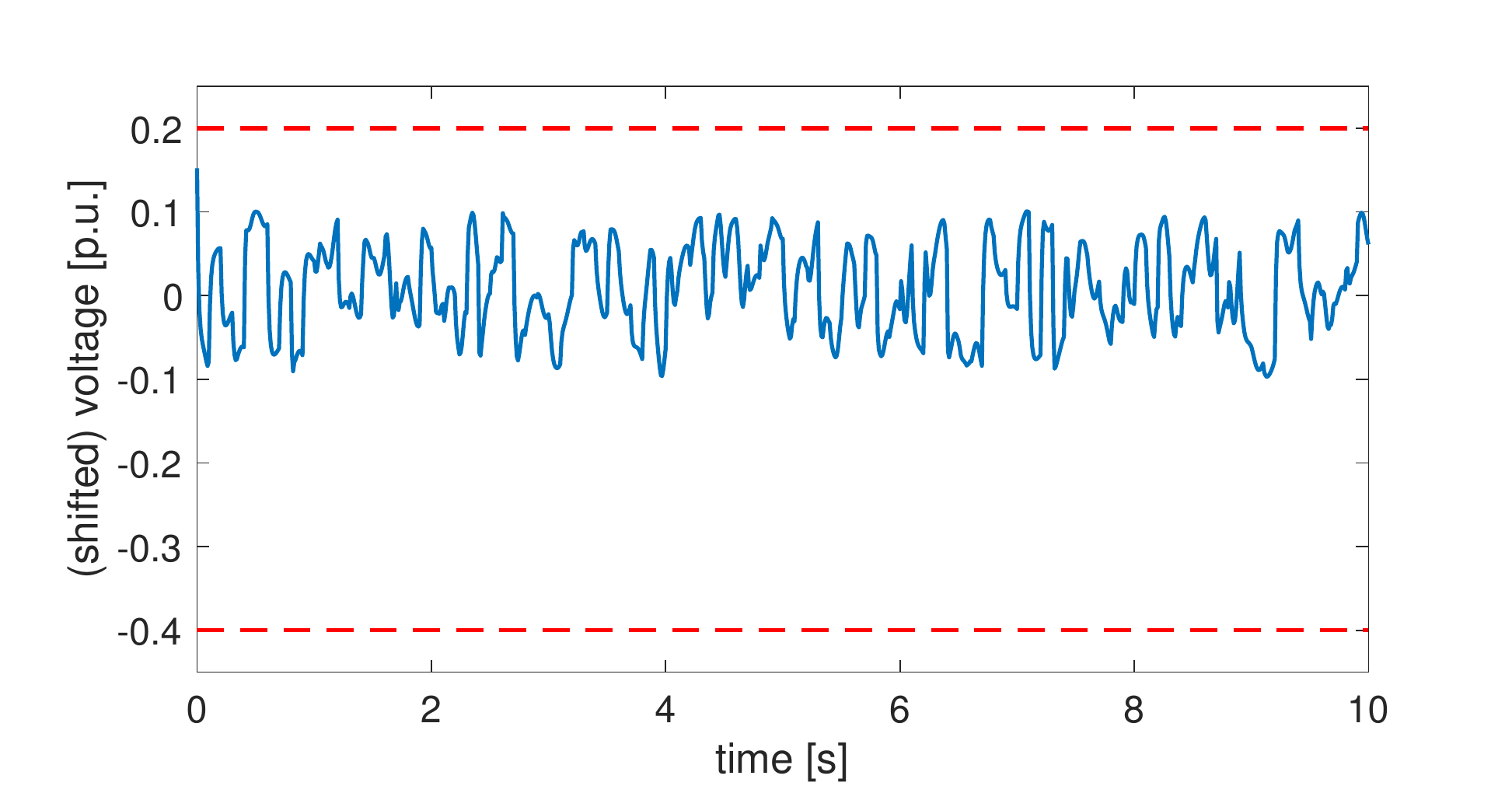}\label{F:res3_vc}
}
\subfigure[]{
\includegraphics[scale=0.4]{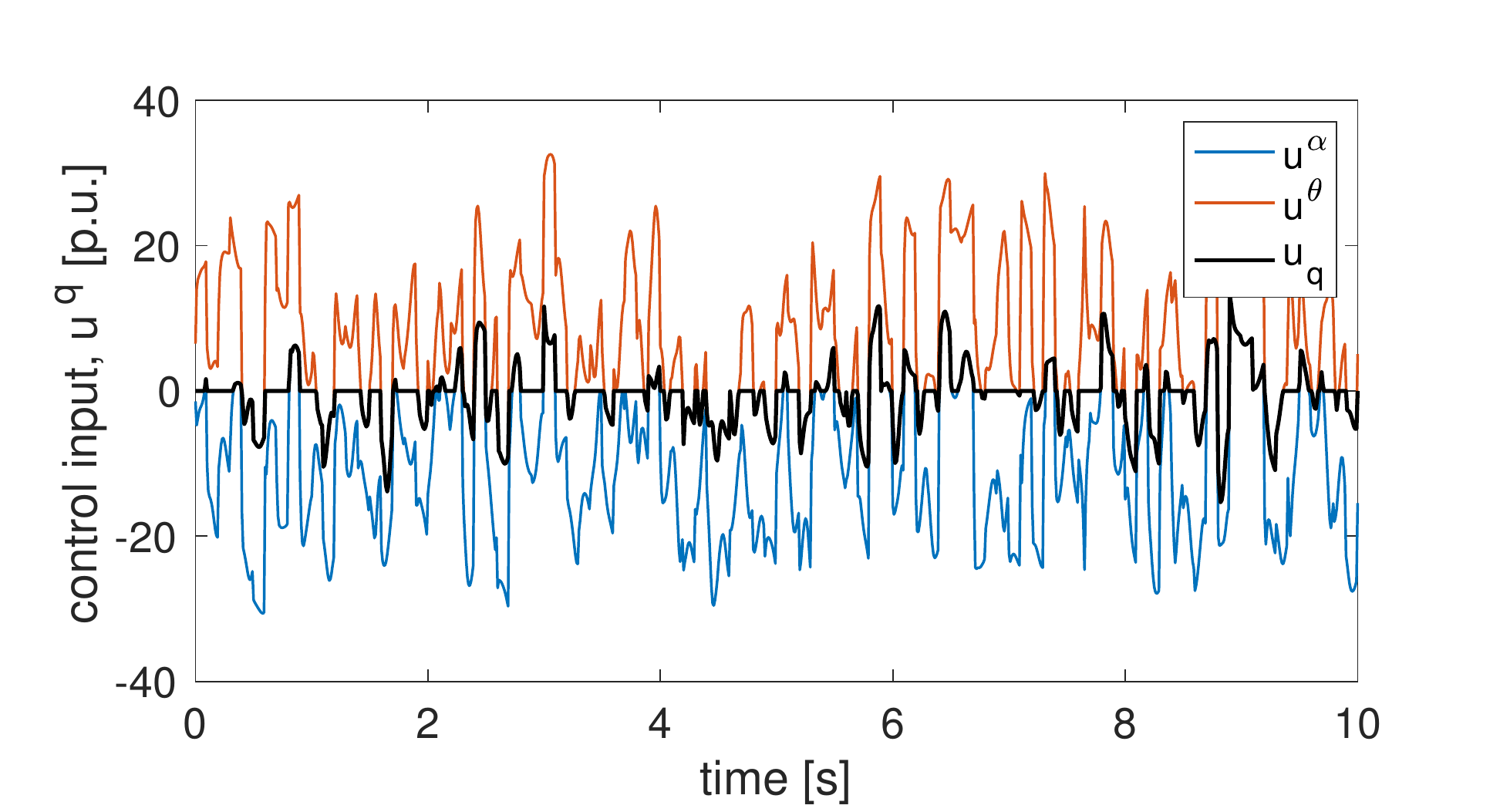}\label{F:res3_uc}
}
\caption[]{Time-series results to illustrate the effect of the controller, under uncertain fluctuations from the neighbors. (a) Dispatched control set-point $u^q=0$ is unchanged. In (b) and (c) the safety filter modifies the control set-point using the computed bounds $u^\alpha$ and $u^\theta$.}
\label{F:res_time}
\end{figure}

\section{Conclusion}

In this paper we consider the problem of safety in inverter-based microgrids. Using barrier functions based methods, we introduce the notion of a safety filter which acts like a buffer between higher level control set-points and the device-level controller. We propose a method to synthesize easy-to-evaluate state-inclusive bounds on the allowable (safe) control inputs. SOS programming was used to design the safety filters in a distributed manner. Numerical results illustrate the effectiveness of the concept. Future work will explore the extension of such methods to larger power systems networks, and explore the integrated operation of these safety filters in a microgrid with a heterogeneous collection of DERs.


\section*{Acknowledgment}

This work was carried out at PNNL (contract DE-AC05-76RL01830) under the support from the U.S. Department of Energy as part of their Grid Modernization Initiative.



\bibliographystyle{IEEEtran}
\bibliography{references,RefKundu,RefBarrier,RefMGStability}

\begin{thebibliography}{10}
\providecommand{\url}[1]{#1}
\csname url@samestyle\endcsname
\providecommand{\newblock}{\relax}
\providecommand{\bibinfo}[2]{#2}
\providecommand{\BIBentrySTDinterwordspacing}{\spaceskip=0pt\relax}
\providecommand{\BIBentryALTinterwordstretchfactor}{4}
\providecommand{\BIBentryALTinterwordspacing}{\spaceskip=\fontdimen2\font plus
\BIBentryALTinterwordstretchfactor\fontdimen3\font minus
  \fontdimen4\font\relax}
\providecommand{\BIBforeignlanguage}[2]{{%
\expandafter\ifx\csname l@#1\endcsname\relax
\typeout{** WARNING: IEEEtran.bst: No hyphenation pattern has been}%
\typeout{** loaded for the language `#1'. Using the pattern for}%
\typeout{** the default language instead.}%
\else
\language=\csname l@#1\endcsname
\fi
#2}}
\providecommand{\BIBdecl}{\relax}
\BIBdecl

\bibitem{pogaku2007modeling}
N.~Pogaku, M.~Prodanovic, and T.~C. Green, ``Modeling, analysis and testing of
  autonomous operation of an inverter-based microgrid,'' \emph{IEEE
  Transactions on power electronics}, vol.~22, no.~2, pp. 613--625, 2007.

\bibitem{taylor2016power}
J.~A. Taylor, S.~V. Dhople, and D.~S. Callaway, ``Power systems without fuel,''
  \emph{Renewable and Sustainable Energy Reviews}, vol.~57, pp. 1322--1336,
  2016.

\bibitem{lasseter2011smart}
R.~H. Lasseter, ``{Smart distribution: Coupled microgrids},'' \emph{Proceedings
  of the IEEE}, vol.~99, no.~6, pp. 1074--1082, 2011.

\bibitem{barklund2008energy}
E.~Barklund, N.~Pogaku, M.~Prodanovic, C.~Hernandez-Aramburo, and T.~C. Green,
  ``Energy management in autonomous microgrid using stability-constrained droop
  control of inverters,'' \emph{IEEE Transactions on Power Electronics},
  vol.~23, no.~5, pp. 2346--2352, 2008.

\bibitem{Xu:2018}
Y.~Xu, C.~Liu, K.~P. Schneider, F.~K. Tuffner, and D.~T. Ton, ``Microgrids for
  service restoration to critical load in a resilient distribution system,''
  \emph{IEEE Transactions on Smart Grid}, vol.~9, no.~1, pp. 426--437, Jan
  2018.

\bibitem{maulik2019stability}
A.~{Maulik} and D.~{Das}, ``Stability constrained economic operation of
  islanded droop-controlled dc microgrids,'' \emph{IEEE Transactions on
  Sustainable Energy}, vol.~10, no.~2, pp. 569--578, April 2019.

\bibitem{simpson2013synchronization}
J.~W. Simpson-Porco, F.~D{\"o}rfler, and F.~Bullo, ``Synchronization and power
  sharing for droop-controlled inverters in islanded microgrids,''
  \emph{Automatica}, vol.~49, no.~9, pp. 2603--2611, 2013.

\bibitem{Schiffer:2014}
J.~Schiffer, R.~Ortega, A.~Astolfi, J.~Raisch, and T.~Sezi, ``Conditions for
  stability of droop-controlled inverter-based microgrids,'' \emph{Automatica},
  vol.~50, no.~10, pp. 2457--2469, 2014.

\bibitem{Vorobev:2018}
P.~Vorobev, P.-H. Huang, M.~Al~Hosani, J.~L. Kirtley, and K.~Turitsyn,
  ``High-fidelity model order reduction for microgrids stability assessment,''
  \emph{IEEE Transactions on Power Systems}, vol.~33, no.~1, pp. 874--887,
  2018.

\bibitem{kundu2019identifying}
S.~{Kundu}, W.~{Du}, S.~P. {Nandanoori}, F.~{Tuffner}, and K.~{Schneider},
  ``Identifying parameter space for robust stability in nonlinear networks: A
  microgrid application,'' in \emph{2019 American Control Conference (ACC)},
  July 2019, pp. 3111--3116.

\bibitem{vorobev2019decentralized}
P.~Vorobev, S.~Chevalier, and K.~Turitsyn, ``Decentralized stability rules for
  microgrids,'' in \emph{2019 American Control Conference (ACC)}.\hskip 1em
  plus 0.5em minus 0.4em\relax IEEE, 2019, pp. 2596--2601.

\bibitem{nandanoori2020distributed}
S.~P. Nandanoori, S.~Kundu, W.~Du, F.~K. Tuffner, and K.~P. Schneider,
  ``Distributed small-signal stability conditions for inverter-based unbalanced
  microgrids,'' \emph{to appear at IEEE Transactions on Power Systems
  (available online at: https://arxiv.org/abs/2003.07979}, 2020.

\bibitem{farhangi2010path}
H.~Farhangi, ``The path of the smart grid,'' \emph{IEEE power and energy
  magazine}, vol.~8, no.~1, 2010.

\bibitem{huang2011future}
A.~Q. Huang, M.~L. Crow, G.~T. Heydt, J.~P. Zheng, and S.~J. Dale, ``The future
  renewable electric energy delivery and management (freedm) system: The energy
  internet.'' \emph{Proceedings of the IEEE}, vol.~99, no.~1, pp. 133--148,
  2011.

\bibitem{kusko2007power}
A.~Kusko and M.~T. Thompson, \emph{Power quality in electrical systems}.\hskip
  1em plus 0.5em minus 0.4em\relax McGraw-Hill, 2007, vol.~23.

\bibitem{Mayne:2000}
D.~Q. Mayne, J.~B. Rawlings, C.~V. Rao, and P.~O. Scokaert, ``Constrained model
  predictive control: Stability and optimality,'' \emph{Automatica}, vol.~36,
  no.~6, pp. 789--814, 2000.

\bibitem{garone2017reference}
E.~Garone, S.~Di~Cairano, and I.~Kolmanovsky, ``Reference and command governors
  for systems with constraints: A survey on theory and applications,''
  \emph{Automatica}, vol.~75, pp. 306--328, 2017.

\bibitem{nicotra2018explicit}
M.~M. {Nicotra} and E.~{Garone}, ``The explicit reference governor: A general
  framework for the closed-form control of constrained nonlinear systems,''
  \emph{IEEE Control Systems Magazine}, vol.~38, no.~4, pp. 89--107, Aug 2018.

\bibitem{Prajna:2007}
S.~Prajna, A.~Jadbabaie, and G.~J. Pappas, ``A framework for worst-case and
  stochastic safety verification using barrier certificates,'' \emph{IEEE
  Transactions on Automatic Control}, vol.~52, no.~8, pp. 1415--1428, 2007.

\bibitem{Wieland:2007}
P.~Wieland and F.~Allg{\"o}wer, ``Constructive safety using control barrier
  functions,'' \emph{IFAC Proceedings Volumes}, vol.~40, no.~12, pp. 462--467,
  2007.

\bibitem{Ames:2017}
A.~D. Ames, X.~Xu, J.~W. Grizzle, and P.~Tabuada, ``Control barrier function
  based quadratic programs for safety critical systems,'' \emph{IEEE
  Transactions on Automatic Control}, vol.~62, no.~8, pp. 3861--3876, 2017.

\bibitem{Wang:2018}
L.~Wang, D.~Han, and M.~Egerstedt, ``Permissive barrier certificates for safe
  stabilization using sum-of-squares,'' \emph{arXiv preprint arXiv:1802.08917},
  2018.

\bibitem{kundu2019distributed}
S.~Kundu, S.~Geng, S.~P. Nandanoori, I.~A. Hiskens, and K.~Kalsi, ``Distributed
  barrier certificates for safe operation of inverter-based microgrids,'' in
  \emph{2019 American Control Conference (ACC)}.\hskip 1em plus 0.5em minus
  0.4em\relax IEEE, 2019, pp. 1042--1047.

\bibitem{chen2019compositional}
Y.~Chen, J.~Anderson, K.~Kalsi, S.~H. Low, and A.~D. Ames, ``Compositional set
  invariance in network systems with assume-guarantee contracts,'' in
  \emph{2019 American Control Conference (ACC)}.\hskip 1em plus 0.5em minus
  0.4em\relax IEEE, 2019, pp. 1027--1034.

\bibitem{zhang2019distributed}
Y.~Zhang and J.~Cortes, ``Distributed bilayered control for transient frequency
  safety and system stability in power grids,'' \emph{arXiv preprint
  arXiv:1906.02861}, 2019.

\bibitem{guerrero2010hierarchical}
J.~M. Guerrero, J.~C. Vasquez, J.~Matas, L.~G. De~Vicu{\~n}a, and M.~Castilla,
  ``{Hierarchical control of droop-controlled AC and DC microgrids - A general
  approach toward standardization},'' \emph{IEEE Transactions on industrial
  electronics}, vol.~58, no.~1, pp. 158--172, 2010.

\bibitem{Coelho:2002}
E.~A.~A. Coelho, P.~C. Cortizo, and P.~F.~D. Garcia, ``Small-signal stability
  for parallel-connected inverters in stand-alone ac supply systems,''
  \emph{IEEE Transactions on Industry Applications}, vol.~38, no.~2, pp.
  533--542, 2002.

\bibitem{Khalil:1996}
H.~K. Khalil, \emph{{Nonlinear Systems}}.\hskip 1em plus 0.5em minus
  0.4em\relax New Jersey: Prentice Hall, 1996.

\bibitem{sostools13}
A.~Papachristodoulou, J.~Anderson, G.~Valmorbida, S.~Prajna, P.~Seiler, and
  P.~A. Parrilo, ``{SOSTOOLS}: Sum of squares optimization toolbox for
  {MATLAB},'' 2013, available from
  \texttt{http://www.eng.ox.ac.uk/control/sostools}.

\bibitem{Putinar:1993}
M.~Putinar, ``Positive polynomials on compact semi-algebraic sets,''
  \emph{Indiana University Mathematics Journal}, vol.~42, no.~3, pp. 969--984,
  1993.

\bibitem{Lasserre:2009}
J.-B. Lasserre, \emph{Moments, Positive Polynomials and Their
  Applications}.\hskip 1em plus 0.5em minus 0.4em\relax World Scientific, 2009,
  vol.~1.

\bibitem{ersal2011impact}
T.~Ersal, C.~Ahn, I.~A. Hiskens, H.~Peng, and J.~L. Stein, ``Impact of
  controlled plug-in evs on microgrids: A military microgrid example,'' in
  \emph{IEEE Power and Energy Society General Meeting}, 2011, pp. 1--7.

\end{thebibliography}
%
%
%

\end{document}